\documentclass[journal]{IEEEtran}
\usepackage{amsmath} 
\usepackage{amssymb}  
\usepackage{tikz}
\usepackage{pdfpages}
\usepackage{mathtools}
\usepackage{amsthm}

\usetikzlibrary{shapes,arrows}
\usetikzlibrary{arrows,automata}
\usetikzlibrary{positioning}

\newtheorem{definition}{Definition}

\newtheorem{corollary}{Corollary}
\newtheorem{theorem}{Theorem}

\ifCLASSINFOpdf
\else
\fi
\begin{document}
%

\title{Reducing Conservatism in Model-Invariant Safety-Preserving Control of Propofol Anesthesia Using Falsification}

\author{Mahdi Yousefi, Klaske van Heusden, Ian M. Mitchell, J. Mark Ansermino and Guy A. Dumont \IEEEmembership{Fellow, IEEE}
\thanks{ M. Yousefi, K. van Heusden and G.A. Dumont are with the Department of Electrical and Computer Engineering, University of British Columbia, Vancouver, BC V6T 1Z4 Canada. (email: $\{$mahdiyou,klaskeh,guyd$\}$@ece.ubc.ca)}
\thanks{I.M. Mitchell is with the Department of Computer Science, University of British Columbia, Vancouver, BC V6T 1Z4 Canada. (email: mitchell@cs.ubc.cs)}
\thanks{J.M. Ansermino is with the Department of Anesthesiology, Pharmacology and Therapeutics, University of British Columbia, Vancouver, BC V6T 1Z3 Canada. (email: anserminos@yahoo.com)}
}%

\maketitle

\begin{abstract}
This work provides a formalized model-invariant safety system for closed-loop anesthesia {\color{black}that uses feedback from measured data for model falsification to reduce conservatism}.
The safety system maintains {\color{black}predicted propofol plasma concentrations, as well as the} patient's blood pressure, within safety bounds despite uncertainty in patient responses to propofol.
{\color{black}Model-invariant formal verification} is used to formalize the safety system.
This technique requires a multi-model description of model-uncertainty.
{\color{black}Model-invariant verification considers all possible dynamics of an uncertain system, and the resulting safety system may be conservative for systems that do not exhibit the worst case dynamical response}.
In this work, we employ model falsification to reduce conservatism of the model-invariant safety system.
Members of a model set that characterizes model-uncertainty are falsified if discrepancy between predictions of those models and measured responses of the uncertain system is established, thereby reducing model uncertainty.
We show that including falsification in a model-invariant safety system reduces conservatism of the safety system.
\end{abstract}


%
\IEEEpeerreviewmaketitle

\section{Introduction}
Reliable operation of safety-critical control systems demands a safety guarantee over the full spectrum of their {\color{black}operational} conditions.
These applications include closed-loop control of physiological variables \cite{van2014safety, murphy2011safety}, flight envelop protection \cite{margellos2009decision}, and process control systems \cite{hasenauer2009guaranteed}.
Formal methods and safety-preserving control techniques provide us with powerful tools to verify safety specifications and design systems to guarantee safety \cite{clarke2008birth}.
These methods can be used to verify feasibility of a set of viability constraints (safe region) for a closed-loop system by calculating the viability kernel \cite{aubin2011viability}.
The viability kernel is a set of states starting from which there exists a control action which satisfies the constraints.
Once feasibility is demonstrated,  a safety-preserving controller can be synthesized to guarantee safety during operation \cite{lygeros1999controllers} \cite{kurzhanski1998ellipsoidal}.

Formal methods and safety-preserving control techniques are model-based analyses and require an accurate model that describes the evolution of system's states over time as a function of the system's inputs. 
These methods have been extended to handle uncertainty; however, the proposed solutions are limited to stochastic and additive state uncertainty \cite{gao2007reachability}, \cite{kaynama2015scalable} and cannot be used for systems with multiplicative model uncertainty.
Although multiplicative uncertainty can be expressed in the form of additive state uncertainty, this may result in very conservative solutions or lead to an empty viability kernel if the level of uncertainty is significant.

Closed-loop anesthesia is an example of safety-critical applications with multiplicative uncertainty.
Closed-loop propofol anesthesia manipulates propofol infusion rates (a commonly used hypnotic agent in intravenous anesthesia) based on feedback of a measured clinical effect on depth of hypnosis to induce and maintain a certain level of anesthesia \cite{dumont2009robust}.
Van Heusden et al. \cite{van2014safety} proposed a safety system for closed-loop anesthesia that maintains the propofol concentration in the plasma and effect-site within the therapeutic window of propofol.
This minimizes the risk of drug under/overdosing.
Safety constraints on other physiological variables such as blood pressure {\color{black}have also been suggested }to improve patient safety during closed-loop anesthesia \cite{khosravi2015constrained}. 
To formalize these safety systems and synthesize a safety-preserving controller, an exact model of each patient is required.
However, patient models are uncertain.
Model uncertainty in closed-loop anesthesia {\color{black} can be represented as multi-model uncertainty, assuming the true patient response is included in the model set}.

In this work, we propose a formalized safety system for closed-loop anesthesia which guarantees that blood pressure of individual patients remains within a safety bound despite model-uncertainty.
We previously introduced a model-invariant verification technique that can be used to synthesize a safety-preserving controller for uncertain systems with multiplicative uncertainty \cite{yousefi2016model}.
Given a multi-model description of model uncertainty, a model-invariant safety-preserving controller satisfies constraints for all members {\color{black}of a model-set}.
{\color{black}Model-invariant safety-preserving control relies on the calculation} of the model-invariant viability kernel.
The model-invariant viability kernel is a set of states starting from which there exists a control input which satisfies safety constraints for all members of the model set.
We showed that the model-invariant viability kernel is the intersection of the viability kernels of all members of the model set.
Conservatism of model-invariant safety preserving control depends on the level of model uncertainty.
As the uncertainty increases, the size of the model-invariant viability kernel decreases.
Moreover, the techniques used for viability kernel approximation may result in further conservatism \cite{kaynama2012computing}, \cite{maidens2013lagrangian}.  

Although safety-preserving control techniques rely on the off-line calculation of the viability kernel and control synthesis, the use of online data has been recommended to improve performance of the formalized safety systems.
For instance, Gillula et al. \cite{gillulay2011guaranteed} paired machine learning algorithms with formal methods and achieved high tracking performance while safety was guaranteed.
In \cite{gillula2013reducing}, they used the same approach to learn disturbances online to reduce conservatism of formalized safety systems.
In this work, we employ online measurements to decrease conservatism of {\color{black}model-invariant safety-preserving control} using model falsification \cite{PKTKN94}.
Given a set of models for an uncertain system, online data can be used to falsify inappropriate members of the model set.
This reduces model uncertainty.
Model falsification inherently deals with missing data and limited excitation.
If falsification criteria are not met, no model will be falsified and safety will not be jeopardized.
In the context of the model-invariant safety preserving control, as members of the model set are falsified, the model-invariant viability kernel can be recalculated online as the intersection of the viability kernels of unfalsified models. 
This yields less conservative model-invariant safety-preserving control if the viability kernels of the falsified models are more restrictive compared to the viability kernel of the true {\color{black}system}.

In this work, we use a model set we identified in \cite{yousefi2017modellingl}, which illustrates the effect of propofol infusion on blood pressure of 10 at-risk patients.
We employ the model-invariant verification technique and formalize a blood pressure safety system for this population.
The model-invariant formalized safety system ensures that blood pressure of the patient remains within a safety bound despite model uncertainty.
In the next step, we employ blood pressure measurement to falsify members of the set to decrease the conservatism. 
Due to the lack of excitation, blood pressure measurement cannot be used to identify the true model of each patient.
However, using clinical data, we show that the blood pressure data can be used to falsify irrelevant models. 

Low blood pressure is common in the period immediately following induction of anesthesia.
However, it is unknown whether closed-loop anesthesia can provide sufficient anesthesia while avoiding hypotension for all patients, especially at-risk patients.
Van Heusden et al. \cite{van2018closed} have recently identified a set of models which relate propofol infusion rates to depth of hypnosis for the same population discussed in \cite{yousefi2017modellingl}.
This model set enables us to study the feasibility of sufficient anesthesia during closed-loop anesthesia in the presence of the blood pressure safety system.
This model set also helps us to better demonstrate conservatism of the formalized model-invariant safety system and the improvement we achieve in terms of performance when falsification is included.
This work demonstrates a proof-of-concept for the blood pressure model-invariant safety system which includes falsification.
The results discussed in this paper are limited to the above-mentioned model set.

The paper is organized as follows: Section \ref{sec:Background} reviews the models used to describe the effects of propofol on patients. 
This section also reviews the existing safety systems proposed in the literature for closed-loop anesthesia.
In Section \ref{sec:falsification}, we describe model falsification and its {\color{black}efficiency} in reducing uncertainty of blood pressure models using clinical data.
In Section \ref{sec:MISPC}, we summarize the results of model-invariant safety-preserving control which was introduced in \cite{yousefi2016model}. 
In Section \ref{sec:Results}, we discuss the feasibility of adequate anesthesia in the presence of the formalized model-invariant blood pressure safety system as well as its conservatism.
Furthermore, we show how model falsification results in reduced conservatism of the model-invariant safety system.
Finally, Section \ref{sec:conclusion} concludes the paper.

\section{Background} \label{sec:Background}
\subsection{Modeling the clinical effect of propofol}
Propofol is an intravenously administered anesthetic drug commonly used in general anesthesia. The relation between propofol and its effect on physiological variables is traditionally described by compartmental pharmacokinetics\footnote{Pharmacokinetics (PK) describes the distribution of drugs in the plasma.} and pharmocodynamics\footnote{Pharmocodynamics (PD) relates drug concentration in the plasma to clinical effects.} (PKPD) models \cite{bibian2005introduction}. The PK model $G_{PK}(s)$ relates the propofol infusion rate $u(t)$ to the plasma concentration $C_p(t)$ :
\begin{align}
\dot{x}(t)&=A_{PK}x(t)+B_{PK}u(t) \nonumber\\
C_p(t)&= C_{PK} x(t).
\label{eq:PK}
\end{align}
Propofol pharmacokinetics are usually described using a three compartment model, corresponding to a state-space representation with three states, where $x_1(t)= C_p(t)$, and $x_2(t)$ and $x_3(t)$ represent the drug concentrations in the fast and slow compartments respectively.
 
The PD model $G_{PD}(s)$ relates the plasma concentration to the drug effect, and typically includes a transfer function relating the plasma concentration to the effect-site concentration $C_e$, 
\begin{equation}
C_e(s) = \frac{k_{e0}}{s+k_{e0}}C_p(s)
\end{equation}
and the nonlinear Hill function relating the effect-site concentration to the clinical effect \cite{absalom108overview}: 
\begin{equation}
E(C_e) = \frac{C_e^\gamma}{EC_{50}^\gamma + C_e^\gamma},
\end{equation}
where $EC_{50}$ denotes the drug effect-site concentration corresponding to 50\% of the maximal drug effect, and $\gamma$ describes the nonlinearity. 
PKPD models are used to guide drug dosing and form the basis of target-controlled infusion (TCI) systems \cite{absalom108overview}. 

Commonly used propofol PK models were identified from data collected from healthy volunteers, e.g. \cite{schnider1998influence}.
Validation and extension to a wider range of patient groups is ongoing, see for example \cite{eleveld2014general}. Pharmacodynamic models are specific to the clinical effect. Most propofol PD models describe the drug effect on the DoH, as used in TCI systems \cite{absalom108overview}. 

Studies quantifying the effect of propofol on blood pressure (BP) are limited \cite{kazama1999comparison,jeleazcov2015pharmacodynamic}. 
Extrapolation based on the identified model dependency on age results in baseline systolic pressure of $223$ mmHg for a 60 year old patient, indicating significant overfitting and limited generalizability of the model identified in \cite{jeleazcov2015pharmacodynamic}. 
Gentilini et al. \cite{gentilini2002} described the effect of alfentanil on BP using a linear model. 

It is known that the infusion rate significantly affects the decrease in arterial pressure during induction of anesthesia \cite{peacock1990effect}. 
The pharmacodynamics on DoH does not vary significantly with age \cite{kazama1999comparison}, while the effect on BP is more pronounced in elderly patients. Based on 12 patients aged 70-85y \cite{kazama1999comparison}, propofol $EC_{50}$ for the drug effect on BP was $\approx 2$ mcg/ml, while $EC_{50}$ for DoH was more than $7$ mcg/ml. 
The pharmacodynamics for DoH were faster than for BP. 
This result implies that maintaining BP while providing sufficient DoH may be challenging for these patients. This study was limited to patients with ASA status I and II, and only propofol was given during the study period. In clinical anesthesia, propofol is administered in combination with an opioid, strongly reducing $EC_{50}$ for DoH. 

These PKPD models describe population average drug responses. For the purpose of robust system development, the response as well as the interpatient variability needs to be quantified, including outlier behaviour \cite{bibian2006patient}. We identified a set of patient models describing the effect of propofol on BP \cite{yousefi2017modellingl} for a subset of the population requiring continuous BP monitoring using an arterial line. Continuous data during induction of anesthesia were available for model identification, providing adequate excitation for model identification. 
In this model set, the blood pressure effect $E_{BP}(t)$ is defined as the percentage of blood pressure decrease from the baseline mean arterial blood pressure ($\overline{BP}(t)$):
{\color{black} 
\begin{equation}
E_{BP}(t)= 100\left(1 - \frac{BP(t)}{\overline{BP}(t)}\right),
\end{equation}}
We identified DoH effect models for the same patient population \cite{van2018closed}, providing a model set describing both the DoH and BP response to propofol infusion, in the presence of remifentanil analgesia.
\subsection{Safety systems for closed-loop propofol anesthesia} 
A robustly designed closed-loop system for propofol anesthesia will provide adequate anesthesia for the patient population considered \cite{dumont2009robust}. 
To ensure safe drug administration, including for extreme outlier behaviour and in the presence of faults, we proposed a safety system that limits drug infusion to bounds within the therapeutic window of propofol \cite{van2014safety}. 
The proposed bounds are expected to be reached for outliers, indicating the extreme patient response to the clinician, and allowing the user to make a clinical decision without compromising patient safety. 

The use of safety constraints on additional physiological variables that are affected by propofol infusion has been proposed to improve system behaviour and safety for at-risk patients \cite{khosravi2015constrained}.
Propofol commonly causes cardiovascular depression and hypotension following induction of anesthesia in manually controlled anesthesia. The pharmacodynamics of the propofol effect on BP are different from the pharmacodynamics of the effect on DoH, both in gain ($EC_{50}$) and speed of response ($k_{e0}$). These differences can be exploited in a safety system. 


The safety systems proposed in \cite{van2014safety} and \cite{khosravi2015constrained} are formalized in \cite{yousefi2017formally} and \cite{yousefi2018formalized}, respectively. 
{\color{black} These safety systems constrain predicted propofol concentrations based on population average models, as measures of these variables are not available in current clinical practice. Individualizing these systems using feedback is challenging since accurate patient models are not available and model uncertainties are significant. Optimal filtering techniques or state observers can therefore not be employed for state estimation as suggested by \cite{yousefi2017output}. }
\subsection{Proposed safety system}\label{sec:Background-3}
This work aims to provide an individualized safety system for blood pressure and formalize it based on the following assumptions:
{\color{black}
\begin{enumerate}
\item The PKPD models of individual patients are unknown.
\item A population-based PK model is available.
\item Individual patient's PKPD models are described by 
$$G_{PKPD}(s)=G_{PK}(s)G_{PD}(s),$$ 
where $G_{PK}(s)$ is known and the blood pressure PD model is known to be part of a finite set of models:
\begin{align}
&G_{PD}(s)\in\{G_i(s)|~i=1,\dots p \},
\end{align}
\item A noisy measure of blood pressure is available.
\item The measurement noise is bounded and an upper limit of this bound is known.
\item The PK states are not measurable.
\item The PK states can be predicted at all times.
\item The BP from all PD models in the finite set can be simulated and predicted at all times.
\end{enumerate}

In this work, we employ the blood pressure PKPD model set identified in \cite{yousefi2017modellingl} to represent uncertainty in the blood pressure response of patients to propofol and formalize a model-invariant safety system \cite{yousefi2016model}. This safety system is augmented using feedback from the (noisy) blood pressure measurements for model falsification. This safety system requires feedback from all patient states. Since the PK states are not measurable, a population based model is formulated based on demographics of each patient to predict the plasma concentration and the other states of the PK model \cite{schuttler2000population}. The state of the PD model can be predicted for all models in the uncertainty set. In addition, the noisy measure of blood pressure available to the safety system is used to falsify members of the PD model set to reduce model uncertainty and consequently conservatism of the safety system.}


\section{Blood-pressure model falsification in propofol anesthesia} \label{sec:falsification}
\subsection{Falsification}
The concept of model falsification was introduced in the context of model validation for robust control \cite{PKTKN94}. Given an {\it{a priori}} nominal model and uncertainty description, the validation problem was cast as a falsification problem; if measured time-domain data is inconsistent with the nominal model and uncertainty bounds, it is falsified (or invalidated). This methodology uses the philosophical principle that a scientific theory can never be proven to be true, but false hypotheses can be falsified by observations \cite{kuhn2012structure}. 

The falsification concept has been used in, for example, data-driven control \cite{ST97, VDS07}, robust adaptive control \cite{Kos01}, and multi-model switching control \cite{BBMT10}. In biomedical applications, controller falsification has been proposed for control of neuromuscular blockade \cite{AMM05}. If a finite number of models/controllers is considered, falsification requires verification of consistency with data for each entry. This approach can be computationally intensive if the initial model/controller space is large or gridding is fine. An analytical approach to controller falsification has been proposed \cite{VDS07} to reduce the computational load and to extend the methodology to an infinite set of candidate controllers. 

\subsection{Blood-pressure model falsification}
{\color{black}In} this section, we formulate a falsification policy which will be used to falsify members of the BP model set identified in \cite{yousefi2017modellingl} based on the following assumptions:
\begin{enumerate}
\item The patient's PK model is known.
\item The patient's BP PD model is unknown.
\item The patient's BP PD model is a member of the BP PD model set:
\begin{align}\label{Eq-Assumption}
&G_{PD}(s)\in\mathcal{M}_{PD}=\{G_i(s)|~i=1,\dots 10 \}.
\end{align}
\item A noisy measure of blood pressure is available.
\item The measurement noise is bounded {\color{black}and an upper limit of this bound is known.} 
\item Outliers are removed from the measurements.
\end{enumerate}

{\color{black}Assume that a patient's measured blood pressure $BP(t)$ is generated as follows: 
\begin{align}
&BP(t) =\mathcal{L}^{-1}\{ G_{PK}(s)G_{PD}(s)u(s) + n(s)\},
\end{align} 
where $u(s)$ is the propofol infusion and $n(s)$ describes the bounded measurement noise.
{\color{black}$\mathcal{L}^{-1}\{\cdot\}$ is the inverse Laplace transform.}
$G_{PK}(s)$ is known, $G_{PD}(s)$ is unknown and satisfies the assumption in equation \eqref{Eq-Assumption}. For each model $G_i(s)$ in the uncertainty set, a corresponding blood pressure can be predicted: 
\begin{align}
\hat{BP}_i(t)=\mathcal{L}^{-1}\{ G_{PK}(s)G_i(s)u(s)\}. 
\end{align} 
The following policy can then be used for model falsification.

\textit{Falsification policy:} $G_i(s)\in\mathcal{M}_{PD}$ is falsified if
\begin{align}
|BP(t)-\hat{BP}_i(t)|>\gamma,
\label{eq:FP}
\end{align}
where
$\gamma$ is defined as
\begin{align}
\gamma= \max\{|n(t)|~|~\forall t\}.
\end{align}

The assumptions that the patient's BP model exists in the model set and $\gamma$ is known guarantee that the patient's true BP model is never falsified.
For the BP PD model set, $\gamma$ was identified as the maximal error between the measurements and model prediction for the 10 patients:
\begin{align}
\gamma =  17\%.
\label{eq:bound}
\end{align}

The assumption that the patient's model is a member of the model set may not be realistic in practice, as no a priori knowledge of the true patient model exists prior to induction of anesthesia.
However, to demonstrate the proof-of-concept, this assumption guarantees that not all members of the model set are falsified.

\begin{figure}[t]
\begin{center}
\includegraphics[scale=.45]{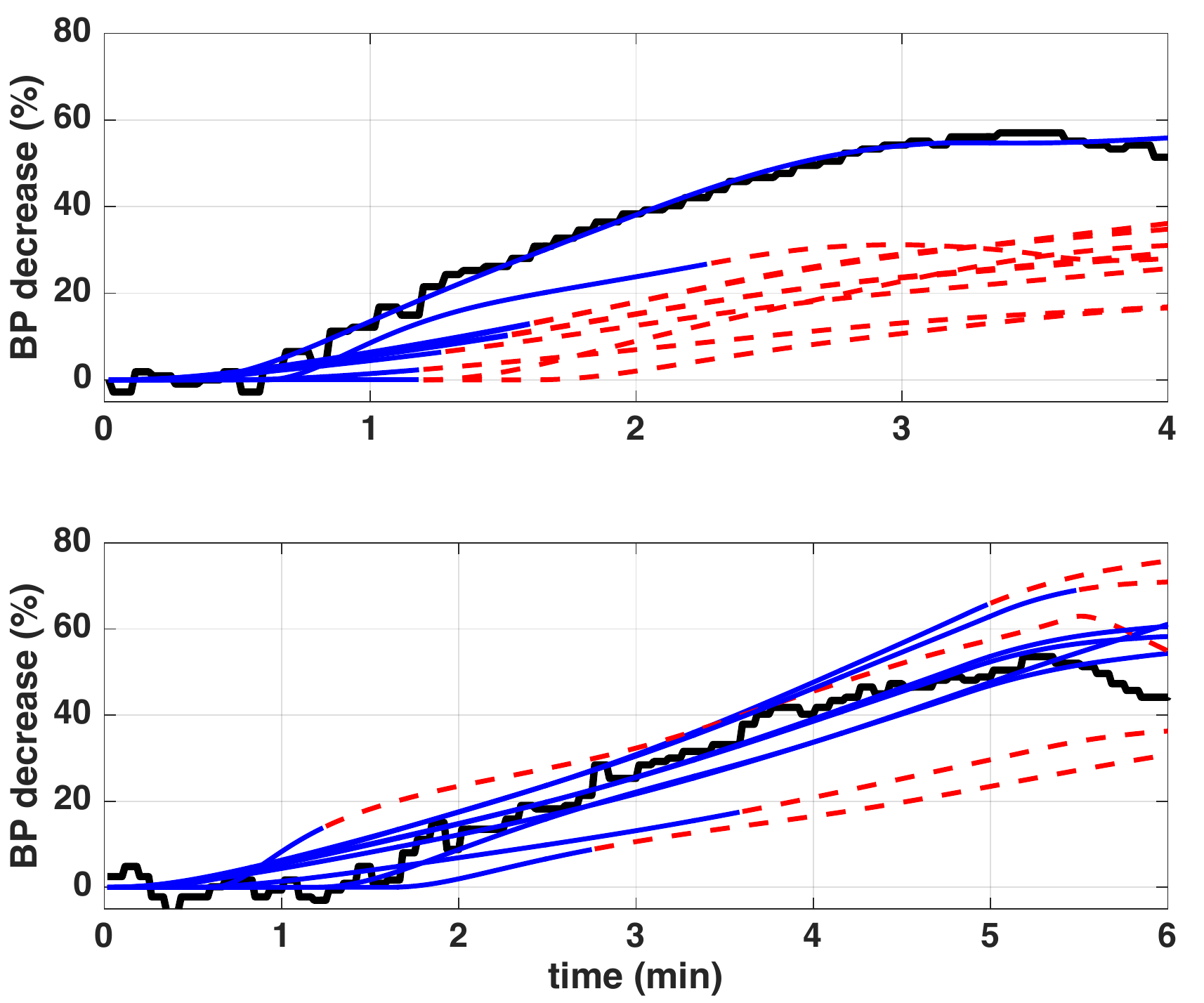}
\end{center}
\caption{Examples of model falsification using clinical data; blue lines: responses of unfalsified models; dashed red line: responses of falsified models; black line: measured BP.} \label{fig:FalsExample}
\end{figure}

Fig. \ref{fig:FalsExample} shows two examples {\color{black}where clinical data was used }to falsify members of the BP model set. According to Fig. \ref{fig:FalsExample}, the first few minutes {\color{black} of measurement data contain sufficient information to falsify outliers.} 
\section{Model-invariant safety-preserving control} \label{sec:MISPC}
\subsection{Safety-preserving control}
Consider a system described by the following state-space model: 
\begin{align}
X:~\dot{x}(t)=Ax(t)+Bu(t), 
\label{eq:ss}
\end{align}
where $x(t)\in\mathbb{R}^n$ and $u(t)\in\mathbb{R}^m$ are the states and input vectors of the system, and $A\in\mathbb{R}^{n\times n}$ and $B\in\mathbb{R}^{m\times n}$ describe the system dynamics. Assume the input is constrained $u(t)\in\mathcal{U}\in\mathbb{R}^m$. 
A formalized safety system ensures that the system states $x(t)$ remain within a (safe) constraint set, $x(t)\in \mathcal{K}\subset\mathbb{R}^n$, over a certain period of time, $t\in[0,\tau]$, where $\mathcal{K}$ and $\mathcal{U}$ are convex compact sets. 
\begin{definition}[Safety-preserving control action]
A control input $u(\cdot):~[t,\tau]\rightarrow \mathcal{U}$ is safety preserving over $[t,\tau]$ if it maintains the states of $X$ within the constraint set $\mathcal{K}$ for all time in $[t,\tau]$.
\end{definition}
The viability kernel describes the system states for which safety-preserving control action exists.
\begin{definition}[Viability kernel]
The finite-horizon viability kernel of $\mathcal{K}$ for the system $X$ is a subset of $\mathcal{K}$ characterizing all states starting from which there exists a constrained control input that maintains the states of $X$ inside $\mathcal{K}$ for all time in $[t,\tau]$:
\begin{align}
&Viab_{[t,\tau]}(\mathcal{K,U},X)=\{x_t\in \mathcal{K}|~x(t)=x_t, ~\exists u(\cdot):[t,\tau]\rightarrow \mathcal{U}\nonumber\\
&s.t.~\forall t'\in[t,\tau],~x(t')\in\mathcal{K} \}.
\end{align}
\end{definition}
If the viability kernel is empty, the safety specification cannot be met with any constrained input. 
In practice, the viability kernel needs to be approximated, see for example \cite{kaynama2012computing}, \cite{maidens2013lagrangian}, \cite{margellos2013viable}. 
If the approximated viability kernel is nonempty, a safety-preserving controller can be found and a controller can be synthesized that provides such safety-preserving control action.

Kurzhanski et al. \cite{kurzhanskiui1997ellipsoidal} showed that starting from any point in $Viab_{[t,\tau]}(\mathcal{K,U},X)$ the following control action $u_{sp}$ is safety-preserving over $[t,\tau]$:
\begin{small}
\begin{align}
&u_{sp}(t)=\arg\min_{u(t)\in\mathcal{U}}\{<l^0\left(x(t),Viab_{[t,\tau]}(\mathcal{K,U},X)\right),Bu(t)>\},
\label{eq:sp}
\end{align}
\end{small}
where
\begin{align}
&l^0\left(x(t),Viab_{[t,\tau]}(\mathcal{K,U},X)\right)=\arg\max_{l}\{<l,x(t)>-\nonumber\\
&\rho(l|Viab_{[t,\tau]}(\mathcal{K,U},X))|~\|l\|_2\leq 1\}.
\label{eq:MaxDir}
\end{align}
$\rho(l|Viab_{[t,\tau]}(\mathcal{K,U},X))$ is the support function of $Viab_{[t,\tau]}(\mathcal{K,U},X)$ in direction $l$:
\begin{align}
&\rho(l|Viab_{[t,\tau]}(\mathcal{K,U},X))=\max\{<l,z>|~z\in\nonumber\\
&Viab_{[t,\tau]}(\mathcal{K,U},X) \}.
\end{align}
In the above equations, $<\cdot,\cdot>:~\mathbb{R}^n\times\mathbb{R}^n\rightarrow \mathbb{R}$ denotes the inner product of vectors, and $\|\cdot\|:~\mathbb{R}^n\rightarrow \mathbb{R}$ is the 2-norm of vectors.
Kurzhanski et al. \cite{kurzhanskiui1997ellipsoidal} proved that the control policy \eqref{eq:sp} is safety-preserving over $t\in[0,\tau]$ if $Viab_{[t,\tau]}(\mathcal{K,U},X)$ is non-empty for all $t$ in $[0,\tau]$.

Kaynama et al. \cite{kaynama2015scalable} combined the above-mentioned safety-preserving control law with a performance controller to preserve safety while meeting performance criteria:
\begin{align}
u(t)=\left\{
\begin{array}{lr}
u_{pr}(t), & x(t)\in \breve{Viab}_{[t,\tau]}(\mathcal{K,U},X),\\
u_{sp}(t), & x(t)\notin \breve{Viab}_{[t,\tau]}(\mathcal{K,U},X).
\end{array}
\right.
\label{eq:uHybrid}
\end{align}
$\breve{Viab}_{[t,\tau]}(\mathcal{K,U},X)$ denotes the interior of ${Viab}_{[t,\tau]}(\mathcal{K,U},X)$.
$u_{pr}(t)$ is a control action provided by a performance controller designed to meet performance criteria, and $u_{sp}(t)$ is calculated based on a safety-preserving control law \eqref{eq:sp}. 
This hybrid control policy guarantees safety while satisfying desired closed-loop performance.

To prevent chattering in the control input, Kaynama et al. \cite{kaynama2015scalable} use a convex combination of $u_{pr}(t)$ and $u_{sp}(t)$ instead of \eqref{eq:uHybrid}:
\begin{align}
u(t) = (1-\zeta)u_{pr}(t)+\zeta u_{sp}(t),
\label{eq:uConvex}
\end{align}
where $\zeta$ is calculated as a function of the distance of $x(t)$ from the boundaries of ${Viab}_{[t,\tau]}(\mathcal{K,U},X)$.

The viability kernel can be calculated offline. The safety preserving control policy can be implemented in real-time, using the pre-calculated viability kernel, the input matrix $B$ and feedback from the full system state $x(t)$. 

\subsection{Model-invariant safety-preserving control} \label{sec:MISafetySys} 
Safety-preserving control as outlined above is a model-based method which requires an accurate system model.
Safety-preserving control of systems with uncertain models and formal verification of such systems have only been discussed in the context of additive and stochastic uncertainties (see \cite{girard2005reachability}, \cite{kaynama2015scalable}, \cite{abate2008probabilistic} and \cite{summers2013stochastic}).
These formal methods are not directly applicable if the system dynamics ($A, B$ in $\eqref{eq:ss}$), are uncertain. While uncertainty can be translated into additive uncertainty, this can introduce significant conservatism if the uncertainty is large. 

We introduced model-invariant safety-preserving control to formalize safety for systems with multiplicative uncertainty \cite{yousefi2016model}. Consider the following multi-model uncertainty description $\mathcal{M}$:
\begin{small}
\begin{align}
\mathcal{M}=\{X_i|X_i:\ &\dot{x}(t)=A_ix(t)+\alpha_i Bu(t),\ \alpha_i >0,\ i=1,\dots,p \}.
 \label{eq:M}
\end{align}
\end{small}
Model-invariant safety-preserving control provides a control action that keeps the states of any of the systems $X_i \in \mathcal{M}$ within the safe set $\mathcal{K}$. 

\begin{definition}[Model-invariant safety-preserving control action]\label{def:u_MISP}
A control action $u(\cdot):[t,\tau]\rightarrow \mathcal{U}$ is model-invariant safety preserving over $[t,\tau]$ if it maintains the states of all members of $\mathcal{M}$ within the constraint set $\mathcal{K}$ for all time in $[t,\tau]$.
\end{definition}
The corresponding model-invariant viability kernel is defined as follows.
\begin{definition}[Model-invariant viability kernel]\label{def:MIVK}
The finite-horizon model-invariant viability kernel of $\mathcal{K}$ for the model set $\mathcal{M}$ is a subset of $\mathcal{K}$ which includes all initial conditions starting from which there exists a constrained control input that maintains the states of all members of $\mathcal{M}$ inside $\mathcal{K}$ for all time in $[t,\tau]$:
\begin{align}
&Viab_{[t,\tau]}(\mathcal{K,U,M})=\{x_t\in \mathcal{K}|~x(t)=x_t,\exists u(\cdot):[t,\tau]\rightarrow \mathcal{U},\nonumber\\
&~s.t.~\forall X_i\in\mathcal{M}~\&~\forall t'\in[t,\tau],~x(t')\in\mathcal{K} \}.
\end{align}
\end{definition}

Let $\mathcal{I}_{[t,\tau]}$ denote the intersection of the viability kernels of all individual set members: 
\begin{align}
\mathcal{I}_{[t,\tau]}=\bigcap_{X_i\in\mathcal{M}}Viab_{[t,\tau]}(\mathcal{K,U},X_i).
\end{align}
The following theorem formulates the main result of \cite{yousefi2016model}: 
\begin{theorem}
Consider the model set $\mathcal{M}$ defined in $\eqref{eq:M}$ and safe region $\mathcal{K}$. 
Assume the intersection of the viability kernels of its individual members $\mathcal{I}_{[t,\tau]}$ is not empty and that $x(t) \in \mathcal{I}_{[t,\tau]}$. Define the following control policy $u_{sp_M}(t)$
\begin{align}
u_{sp_M}(t)=\arg\min_{u(t)\in\mathcal{U}}\{<l^0\left(x(t),\mathcal{I}_{[t,\tau]}\right),Bu(t)> \},
\end{align}
where $l^0\left(x(t),\mathcal{I}_{[t,\tau]}\right)$ is defined as \eqref{eq:MaxDir}.
Then the control policy $u_{sp_M}(t)$ is model-invariant safety-preserving.  
\end{theorem}
\begin{proof}
For ${X}_j\in\mathcal{M}$, let define
\begin{align}
V(t)=Dist^2\left(x(t),Viab_{[t,\tau]}(\mathcal{K,U},X_j)\right) ,
\end{align}
where $Dist\left(x(t),Viab_{[t,\tau]}(\mathcal{K,U},X_j)\right)$ is the Hausdorff distance measuring the distance of $x(t)$ from $Viab_{[t,\tau]}(\mathcal{K,U},X_j)$ \cite{kurzhanskiui1997ellipsoidal}:
\begin{align}
&Dist\left(x(t),Viab_{[t,\tau]}(\mathcal{K,U},X_j)\right)=\nonumber\\
&\min\{\|x(t)-v\|_2|~v\in Viab_{[t,\tau]}(\mathcal{K,U},X_j)\}=\nonumber\\
&\max\{<l,x(t)>-\rho(l|Viab_{[t,\tau]}(\mathcal{K,U},X_j))|~\|l\|_2 \leq1\}=\nonumber\\
&<l^0\left(x(t),Viab_{[t,\tau]}(\mathcal{K,U},X_j)\right),x(t)>-\nonumber\\
&\rho(l^0\left(x(t),Viab_{[t,\tau]}(\mathcal{K,U},X_j)\right)|Viab_{[t,\tau]}(\mathcal{K,U},X_j)).
\end{align}
Kurzhanski et al. \cite{kurzhanskiui1997ellipsoidal} showed that assuming $\dfrac{d}{dt}V(t)$ exists, starting form any point in $Viab_{[t,\tau]}(\mathcal{K,U},X_j)$ the following control policy is safety preserving over $[t,\tau]$ for $X_j$:
\begin{align}
u(t)=\arg\min_u\{\dfrac{d}{dt}V(t)|~u\in\mathcal{U}\}.
\label{eq:u_sp1}
\end{align}
Due to the fact that $Dist\left(x(t),Viab_{[t,\tau]}(\mathcal{K,U},X_j)\right)\geq 0$, \eqref{eq:u_sp1} can be expressed as
\begin{align}
u(t)=\arg\min_{u(t)\in\mathcal{U}}\{\dfrac{d}{dt}Dist\left(x(t),Viab_{[t,\tau]}(\mathcal{K,U},X_j)\right)\}.
\label{eq:u_sp2}
\end{align}
Since $\mathcal{I}_{[t,\tau]} \subseteq Viab_{[t,\tau]}(\mathcal{K,U},X_j)$, we can use $\mathcal{I}_{[t,\tau]}$ in \eqref{eq:u_sp2} as an under-approximation of $Viab_{[t,\tau]}(\mathcal{K,U},X_j)$ and rewrite the safety-preserving control policy \eqref{eq:u_sp2} as:
\begin{align}
u(t)=\arg\min_u\{\dfrac{d}{dt}Dist\left(x(t),\mathcal{I}_{[t,\tau]}\right)|~u\in\mathcal{U}\}.
\label{eq:u_sp3}
\end{align}
Accordingly, starting from any point in $\mathcal{I}_{[t,\tau]}$, \eqref{eq:u_sp3} is safety preserving over $[t,\tau]$ for $X_j$.
The derivative of the Hausdorff distance can be simplified as
\begin{align}
&\dfrac{d}{dt}Dist\left(x(t),\mathcal{I}_{[t,\tau]}\right)=
<l^0\left(x(t),\mathcal{I}_{[t,\tau]}\right),\dot{x}(t)> \nonumber\\
&-\dfrac{\partial}{\partial t}\rho(l^0\left(x(t),\mathcal{I}_{[t,\tau]}\right)|\mathcal{I}_{[t,\tau]}).
\label{eq:d_dist}
\end{align}
According to \cite{kurzhanskiui1997ellipsoidal} and due to the fact $\mathcal{I}_{[t,\tau]}\subseteq Viab_{[t,\tau]}(\mathcal{K,U},X_j)$, the partial derivative of the support function can be written as
\begin{align}
&\dfrac{\partial}{\partial t}\rho(l^0\left(x(t),\mathcal{I}_{[t,\tau]}\right)|\mathcal{I}_{[t,\tau]})=<l^0\left(x(t),\mathcal{I}_{[t,\tau]}\right),A_jx(t)>+\nonumber\\
&\rho(l^0\left(x(t),\mathcal{I}_{[t,\tau]}\right)|\alpha_jB\mathcal{U}).
\label{eq:d_SF}
\end{align}
Substituting \eqref{eq:d_SF} in \eqref{eq:d_dist} yields
\begin{align}
&\dfrac{d}{dt}Dist\left(x(t),\mathcal{I}_{[t,\tau]}\right)=\nonumber\\
&<l^0\left(x(t),\mathcal{I}_{[t,\tau]}\right),A_jx(t)+\alpha_jBu(t)>-\nonumber\\
&<l^0\left(x(t),\mathcal{I}_{[t,\tau]}\right),A_jx(t)>- \rho(l^0\left(x(t),\mathcal{I}_{[t,\tau]}\right)|\alpha_jB\mathcal{U})=\nonumber\\
&<l^0\left(x(t),\mathcal{I}_{[t,\tau]}\right),\alpha_jBu(t)>-\rho(l^0\left(x(t),\mathcal{I}_{[t,\tau]}\right)|\alpha_jB\mathcal{U}).
\end{align}
Consequently, we can express \eqref{eq:u_sp3} as
\begin{align}
u(t)=\arg\min_{u(t)\in\mathcal{U}}\{<l^0\left(x(t),\mathcal{I}_{[t,\tau]}\right),\alpha_jBu(t)>\}.
\label{eq:u_sp4}
\end{align}
Due to the fact that $\alpha_i>0$ and \eqref{eq:u_sp4} is convex, the minimizer of \eqref{eq:u_sp4} is independent of $\alpha_i$.
Thus, we can simplify \eqref{eq:u_sp4} to the following optimization problem:
\begin{align}
u(t)=\arg\min_{u(t)\in\mathcal{U}}\{<l^0\left(x(t),\mathcal{I}_{[t,\tau]}\right),Bu(t)>\}.
\label{eq:u_sp5}
\end{align} 
Since $\mathcal{I}_{[t,\tau]}\subseteq Viab_{[t,\tau]}(\mathcal{K,U},X_i)$ for all $X_i\in\mathcal{M}$, \eqref{eq:d_SF} holds for all members of $\mathcal{M}$ and \eqref{eq:u_sp3} simplifies to \eqref{eq:u_sp5} for all $X_i\in\mathcal{M}$.
Thus, \eqref{eq:u_sp5} maintains the states of all members of $\mathcal{M}$ over $[0,\tau]$ within $\mathcal{K}$.
Consequently, according to \textit{Definition} \ref{def:u_MISP}, \eqref{eq:u_sp5} is model-invariant safety preserving.
\end{proof}
It follows by definition that the intersection is a subset of the model-invariant viability kernel:
\begin{align}
\mathcal{I}_{[t,\tau]} \subseteq Viab_{[t,\tau]}(\mathcal{K,U,M}).
\label{eq:IsubViab}
\end{align}
Note that $B$ in $\eqref{eq:M}$ is known, $\mathcal{I}_{[t,\tau]}$ can be calculated offline, the model-invariant safety-preserving control policy $u_{sp_M}(t)$ can be implemented in real-time using feedback from the measured state $x(t)$, and the same model-invariant safety-preserving controller can be used for any system $X_i\in\mathcal{M}$.
\begin{corollary}
\begin{align}
Viab_{[t,\tau]}(\mathcal{K,U,M})=\mathcal{I}_{[t,\tau]}.
\label{eq:IeqViab}
\end{align}
\end{corollary}
\begin{proof}
According to \eqref{eq:IsubViab}, to show \eqref{eq:IeqViab} holds, it is sufficient to show that 
\begin{align}
Viab_{[t,\tau]}(\mathcal{K,U,M})\subseteq \mathcal{I}_{[t,\tau]}.
\label{eq:ViabsubI}
\end{align}
Suppose \eqref{eq:ViabsubI} does not hold:
\begin{align}
Viab_{[t,\tau]}(\mathcal{K,U,M})\not\subseteq \mathcal{I}_{[t,\tau]}.
\label{eq:ViabnsubI}
\end{align}
Accordingly, 
\begin{align}
\exists x'\in Viab_{[t,\tau]}(\mathcal{K,U,M}) ~s.t.~ x'\not\in\mathcal{I}_{[t,\tau]}.
\label{eq:crl2}
\end{align}
Due to the fact that $\mathcal{I}_{[t,\tau]}=\bigcap_{X_i\in\mathcal{M}}Viab_{[t,\tau]}(\mathcal{K,U},X_i)$, \eqref{eq:crl2} implies that
\begin{align}
\exists X_j\in\mathcal{M}~s.t.~x'\not\in Viab_{[t,\tau]}(\mathcal{K,U},X_j).
\end{align}
Therefore, starting from $x'$ there is no safety-preserving control action to preserve safety of $X_j$ over $[t,\tau]$.
According to the definition of the model-invariant viability kernel (\textit{Definition} \ref{def:MIVK}), $x'$ cannot be a member of $Viab_{[t,\tau]}(\mathcal{K,U,M})$.
This means
\begin{align}
\not\exists x'\in Viab_{[t,\tau]}(\mathcal{K,U,M})~s.t.~x'\not\in \mathcal{I}_{[t,\tau]}.
\label{eq:crl1}
\end{align}
Equation \eqref{eq:crl1} contradicts the assumption that \eqref{eq:ViabsubI} does not hold.
Therefore, \eqref{eq:ViabsubI} holds by contradiction.
Comparing \eqref{eq:IsubViab} and \eqref{eq:ViabsubI} yields
\begin{align}
Viab_{[t,\tau]}(\mathcal{K,U,M})=\mathcal{I}_{[t,\tau]}.
\end{align}
\end{proof}
\section{Falsified robust safety-preserving control of propofol anesthesia} \label{sec:Results}
While {\it{a priori}} information is insufficient to establish which individual model $X_i$ is controlled, data collected during operation can be used to reduce the uncertainty and consequently the conservatism introduced by model-invariant safety-preserving control. 

This scheme takes advantage of the characteristics of falsification, where limited excitation and missing data are dealt with naturally. If the data does not contain sufficient information to determine consistency with the model, it cannot be falsified. While in this situation falsification may not reduce conservatism, it does not affect the safety of model-invariant safety-preserving control. When data is missing due to for example sensor faults, there is no information to determine consistency and again, no model can be falsified. These characteristics are particularly important in biomedical applications, where excitation is limited and missing data is not uncommon. 

{\color{black}To illustrate the effectiveness of the proposed method, closed-loop anesthesia is simulated using three safety systems: 
\begin{enumerate}
\item Closed-loop anesthesia with an individualized safety system. This safety system formalizes individualized safety constraints assuming the patient models are known. This scheme results in minimum conservatism as we assume the patient models are known and there is no model uncertainty, however, this assumption is not realistic in practice. 
\item Closed-loop anesthesia with a model-invariant safety system. This safety system guarantees safety for all patient responses included in the model set, can be implemented in practice, but introduces significant conservatism. 
\item Closed-loop anesthesia with a model-invariant safety system using falsification. This formalized safety system includes feedback from the measured blood pressure response to falsify models, as proposed in this paper. This system can be implemented in practice and it is shown that it significantly reduces conservatism. 
\end{enumerate}

The recommended range of the DoH index during general anesthesia is 40-60 \cite{agrawal2010recommended}.
The index of 50 is mostly targeted in closed-loop anesthesia \cite{zikov2006quantifying}.
However, it is unknown if the index of 50 is achievable for all patients including at-risk patients in the presence of blood pressure constraints.
In this work, we study feasibility of sufficient anesthesia (DoH index of 50) in the presence of various blood pressure constraints.
}
\subsection{Simulation results}

\subsubsection{Closed-loop anesthesia with an individualized safety system}
Here, we formalize a blood pressure safety system for each patient assuming:
\begin{enumerate}
\item The BP (and DoH) PKPD model is known for each patient.
\item All states of the PKPD model are measurable {\color{black}or predicted}.
\item There is no measurement noise.
\end{enumerate}
The above assumptions are sufficient to formalize an individualized safety system with minimum conservatism.
However, existing methods to approximate the viability kernel can result in a certain level of conservatism.
Maiden et al. in \cite{maidens2013lagrangian} compared conservatism of these methods.
In this paper, we employ convex polytopes to represent constraint sets.
We use the Multi-Parametric Toolbox 3.0 \cite{MPT3} to conduct operations on polytopes to calculate the viability kernel.
We use the the recursive approach developed in \cite{maidens2013lagrangian} to approximate the viability kernel.

We employ the set of 9 BP PKPD models identified in \cite{yousefi2017modellingl} for which van Heusden et al. \cite{van2018closed} identified DoH PKPD models.
We use the BP PKPD model of each patient to formalize a blood pressure safety system for that specific patient.
We employ the corresponding DoH PKPD models to simulate the response of each patient during closed-loop anesthesia and in the presence of the BP safety system.
We employ the safety constraints proposed in \cite{van2014safety} which limit the PK states to the therapeutic window of propofol:
\begin{align}
C_p, x_2,x_3 \in [0,10\text{mg/l}].
\label{eq:PKConstraints}
\end{align}
The following constraint on the infusion rate of propofol is also suggested in \cite{van2014safety}:
\begin{align}
u(t)\in[0,600 \text{ml/h}].
\label{eq:inputConstraints}
\end{align}
\begin{figure}[t]
\begin{center}
\includegraphics[scale=.4]{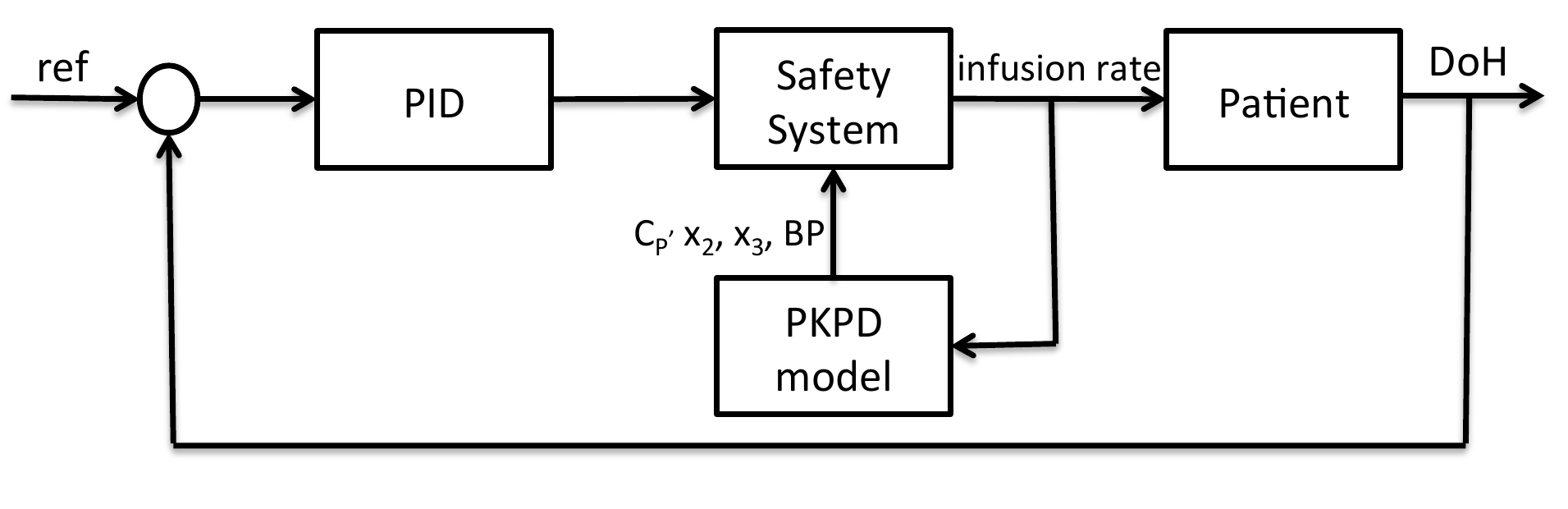}
\end{center}\caption{Block diagram of the closed-loop anesthesia the individualized safety system.}
\label{fig:BlockDiag_IndivSS}
\end{figure}
\begin{figure}[t]
\begin{center}
\includegraphics[scale=.45]{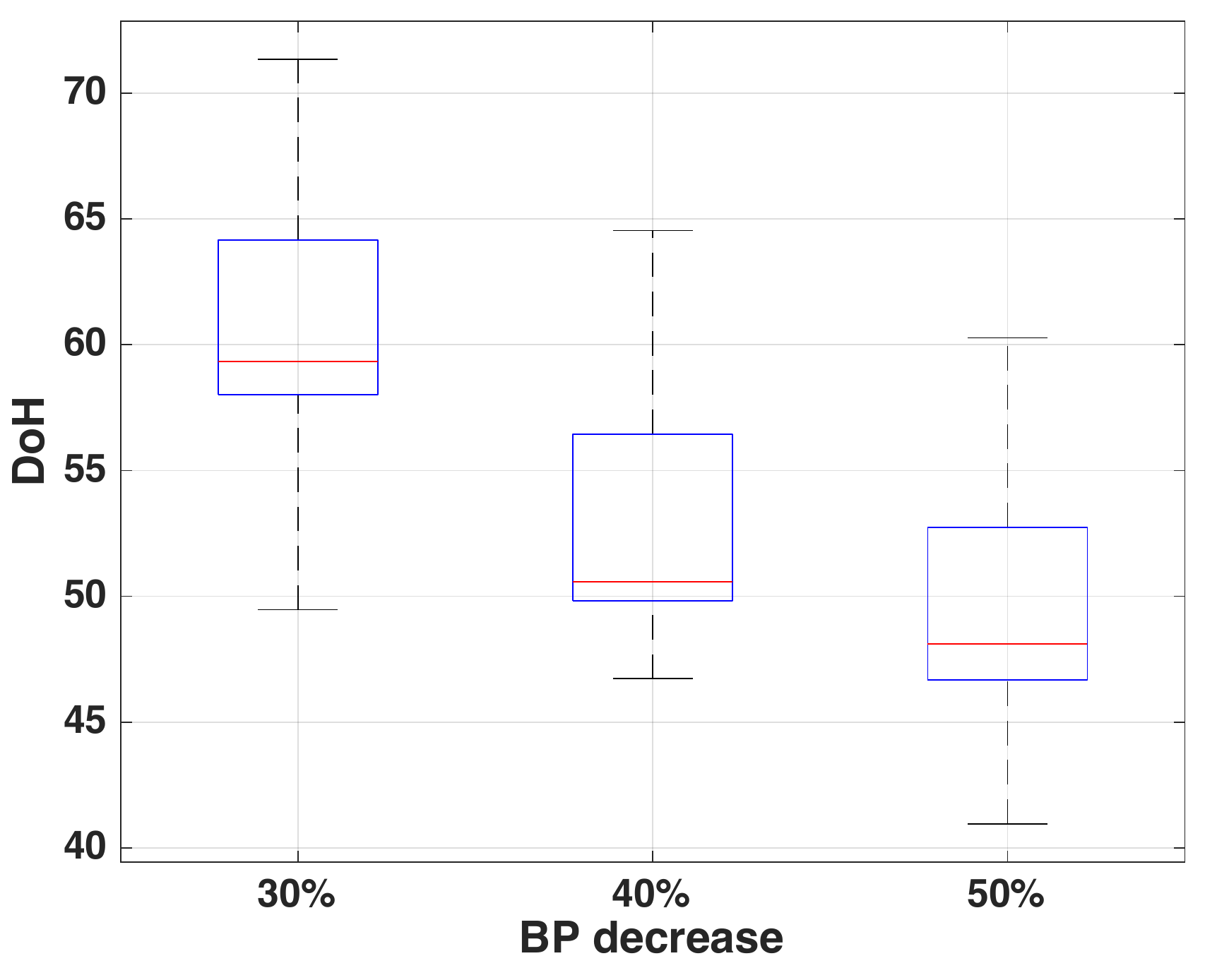}
\end{center}\caption{Comparing DoH of the patients with different constraints on BP, achieved after 20 min from the start of the closed-loop administration of propofol with the individualized safety system.}
\label{fig:CMPBPDrops}
\end{figure}

We discuss three different constraints on BP:
\begin{align}
\text{BP drop} \leq 30\%,~40\%,~50\%.
\end{align}

We employ the PID controller robustly tuned by Dumont et al. \cite{dumont2011closed} to achieve the closed-loop goal.
We include back-calculation anti-windup suggested by van Heusden et al. \cite{van2014safety} to improve the performance of the PID controller when the safety constraints are active.
For each patient, we formalize the safety system to guarantee that the states of the PKPD model remain within the safety constraints.
The viability kernels of the patient models are given in \cite{yousefi2017modellingl}.
Fig. \ref{fig:BlockDiag_IndivSS} illustrates the block diagram of the implemented system in simulation. 

{\color{black}Fig. \ref{fig:CMPBPDrops} shows the box plot of the DoH index of patients with different constraints on BP decrease at $t=20\text{min}$ following the start of propofol infusion. 
Induction of anesthesia cannot be completed for all cases when the BP decrease is limited to 30\% and 40\% \footnote{The induction of anesthesia is completed if DoH goes below 60 and stays there for more than 30s.}.
In contrast, when the BP decrease is limited to 50\%, induction of anesthesia can be completed in all cases.
Fig. \ref{fig:ResponsesBPls50} shows the closed-loop responses of the patients with the bound on BP decrease at 50\%.}
\begin{center}
\begin{figure}[t]
\includegraphics[scale=.45]{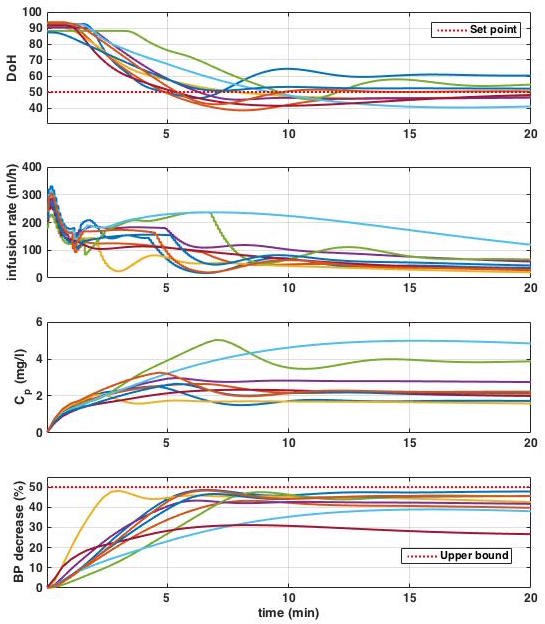}
\caption{The closed-loop responses of the patients with BP decrease limited to 50\% (individualized safety system).}
\label{fig:ResponsesBPls50}
\end{figure}
\end{center}
\subsubsection{Closed-loop anesthesia with a model-invariant safety system}
In the previous section, we showed that all  patients could achieve a DoH index between 40 and 60 when 50\% decrease in blood pressure is allowed. In this section, we formalize a model-invariant viability kernel for the same population based on the assumption discussed in section \ref{sec:Background-3}.

We assume the PK model of each patient is known \cite{schuttler2000population}.
We assume the patient's BP PD model is unknown, however it is part of the model set identified in \cite{yousefi2017modellingl}. 
{\color{black}We calculate the model-invariant viability kernel according to this multi-model uncertainty description.}

We limit blood pressure decrease to be less than $50\%$ and use the constraints defined in \eqref{eq:PKConstraints} and \eqref{eq:inputConstraints} on the states of the PK model and the infusion rate.
The calculated model-invariant viability kernel under the mentioned assumptions is given in \cite{yousefi2017modellingl}.
{\color{black}The resulting model-invariant safety-preserving controller requires feedback from all states of the BP model. The states of the patient PK model are assumed known, and can be predicted. The PD model on the other hand is uncertain, unknown, and cannot be predicted for each patient. 

Although we assume a noisy measure of blood pressure is available, this measurement can not be used as feedback to the model-invariant safety-preserving control.
A limited number of solutions have been proposed in literature to deal with measurement noise in safety-preserving control schemes (e.g. \cite{yousefi2017output},\cite{lesser2016safety}).}
These solutions rely on formulating a state-observer which requires a known model of a system.
To the best of our knowledge, safety-preserving control in the presence of measurement noise as well as model-uncertainty has not been discussed in literature.
\begin{figure}[t]
\begin{center}
\includegraphics[scale=.4]{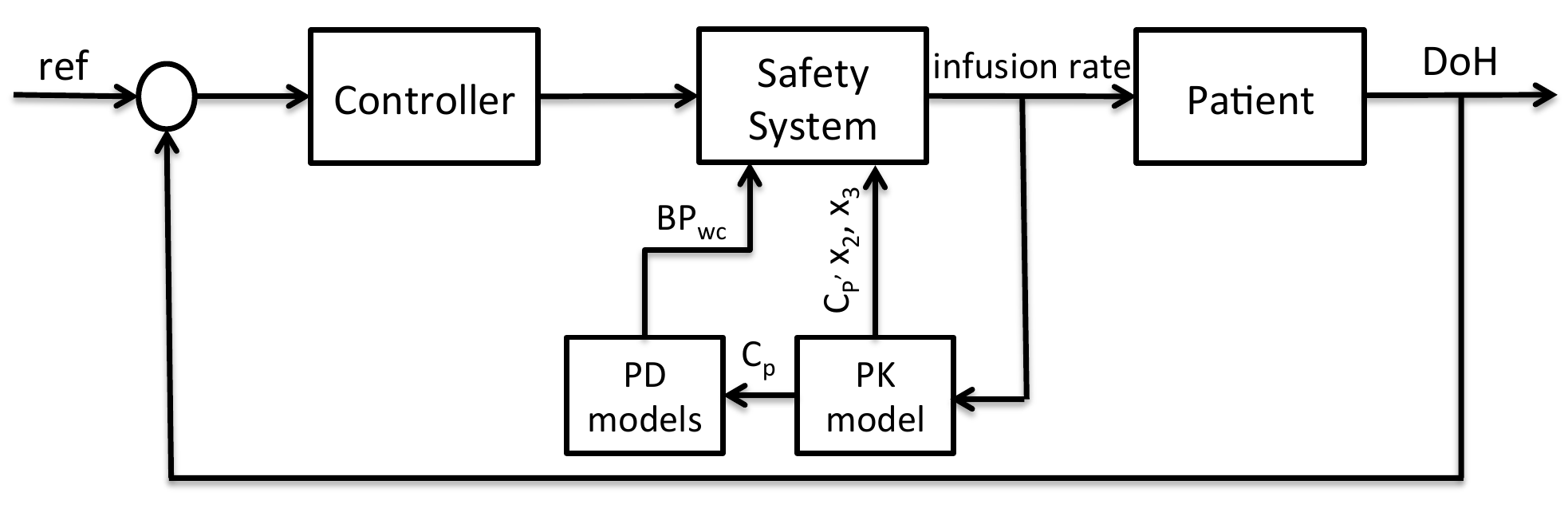}
\end{center}\caption{Block diagram of the closed-loop anesthesia with the model-invariant safety system.}
\label{fig:BlockDiag_MISP}
\end{figure}

In this paper, {\color{black} we therefore employ} the worst case prediction of blood pressure decrease as feedback to the model-invariant safety-preserving control.
We define the worst case blood pressure decrease $BP^{\downarrow}_{wc}$ as follows:
\begin{align}
BP^\downarrow_{wc}=\arg\max_{BP^\downarrow_i}\{~|BP^\downarrow_i|~|~i=1,\dots,9\},
\label{eq:WCBP}
\end{align}
where $BP^\downarrow_i$ is the blood pressure decrease predicted using the $i^{th}$ member of the BP model set.
The model-invariant safety-preserving control maintains the states of all members of the model set within the safe region.
According to the definition of the worst case blood pressure decrease, the control input which keeps $BP^{\downarrow}_{wc}$ less than $50\%$, maintains all $BP^\downarrow_i$s below $50\%$.
Consequently, since we assume that the true blood pressure model of the patients exists in the model set, maintaining the worst case predicted blood pressure decrease below $50\%$ guarantees that the decrease in blood pressure of each patient stays within the safe range.
Fig. \ref{fig:BlockDiag_MISP} illustrates the block diagram of the implemented system in simulation. 
\begin{figure}[t]
\begin{center}
\includegraphics[scale=.45]{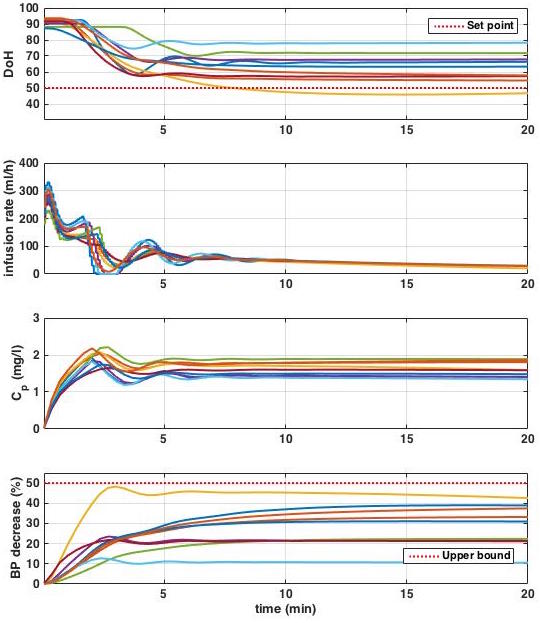}
\end{center}
\caption{The closed-loop responses of the patients with BP decrease limited to 50\% (model-invariant safety system).}
\label{fig:MISP}
\end{figure}

{\color{black}The bottom plot in Fig. \ref{fig:MISP} shows blood pressure decrease of all patients using the model-invariant safety system. 
The safety system maintains the blood pressure of all patients within the safe bound. 
However, the safety system does not allow the closed-loop controller to complete induction of anesthesia for the majority of the population.
The model-invariant safety system introduces significant conservatism compared to the individualized safety system (see Fig. \ref{fig:CMPSP-MISP-Fals_20min}).}
\subsubsection{Closed-loop anesthesia with a model-invariant safety system using falsification}
{\color{black}The model-invariant viability kernel is calculated as the intersection of the viability kernels of all individual members of the model set. 
Accordingly, if the viability kernel of one of the models is restrictive, the model-invariant safety system restricts the drug infusion for all patients. In Section \ref{sec:falsification}, we showed that clinical blood pressure data during the first minutes following the start of propofol infusion contain sufficient information to falsify outliers. By falsifying restrictive models in the model set, and recalculating the model-invariant viability kernel as the intersection of the remaining models, conservatism of model-invariant safety-preserving control can be reduced. If no model is falsified, the model-invariant viability kernel remains unchanged and safety is not compromised.}

Here, we employ the model-invariant safety system we formalized in the previous section, but instead of calculating the model invariant viability kernel off-line, we calculate it online.
At each sample time, we falsify a member of the model set if the difference between its simulated response and the measured BP is bigger than a certain threshold.
Then, we calculate the model-invariant viability kernel (the intersection) with the viability kernel of the falsified model removed.
We specify the threshold according to equation (\ref{eq:bound}).
Furthermore, the worst case predicted blood pressure decrease, which is defined in \eqref{eq:WCBP}, is selected from the predictions of the unfalsified models.
Fig. \ref{fig:BlockDiag_MISP_Fals} illustrates the block diagram of the implemented closed-loop anesthesia with the model-invariant safety system with falsification. 

Fig. \ref{fig:Results_WCBPls50_MISS_Fals} illustrates the closed-loop responses of the patients when the model-invariant safety system with falsification is in place.
Accordingly, induction of anesthesia is completed in all cases and DoH reaches an index between 40-60 for all patients.
Comparing Fig. \ref{fig:MISP} and Fig, \ref{fig:Results_WCBPls50_MISS_Fals} depicts a significant improvement in the performance of the model-invariant safety system when falsification is included.
Fig. \ref{fig:CMPSP-MISP-Fals_20min} compares conservatism of the safety systems discussed in this work.
The individualized safety system maintains safety with minimum conservatism while the model-invariant safety system significantly increases conservatism.
However, adding model falsification to the model-invariant safety system significantly decreases conservatism of the model-invariant safety system and brings it to the conservatism level of the individualized safety system.
\begin{figure}[t]
\begin{center}
\includegraphics[scale=.4]{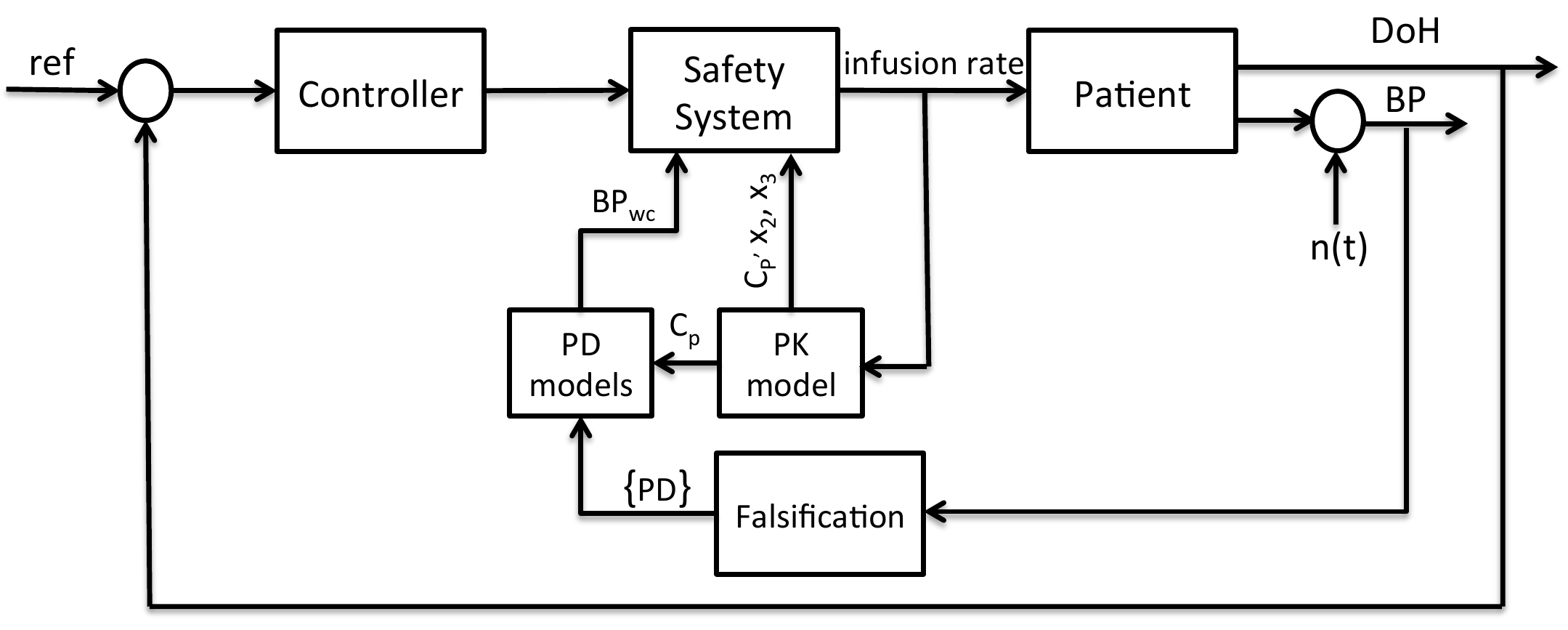}
\end{center}\caption{Block diagram of the closed-loop anesthesia with the model-invariant safety system including falsification.}
\label{fig:BlockDiag_MISP_Fals}
\end{figure}

\begin{center}
\begin{figure}[t]
\includegraphics[scale=.45]{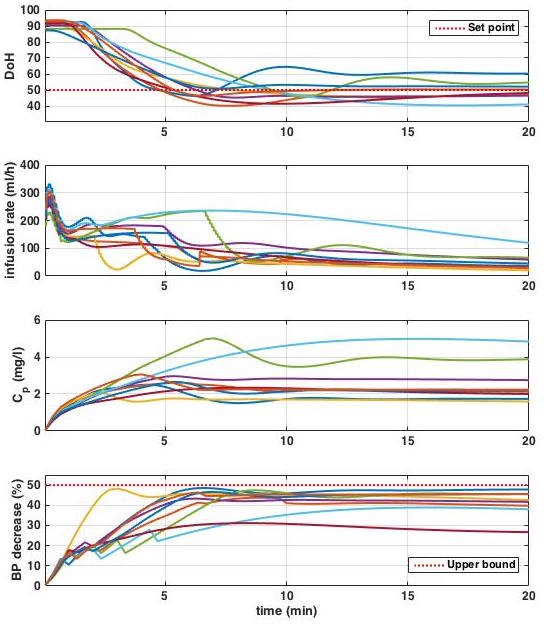}
\caption{The closed-loop responses of the patients with BP decrease limited to 50\% (individualized safety system including falsification).}
\label{fig:Results_WCBPls50_MISS_Fals}
\end{figure}
\end{center}

\begin{center}
\begin{figure}[t]
\includegraphics[scale=.45]{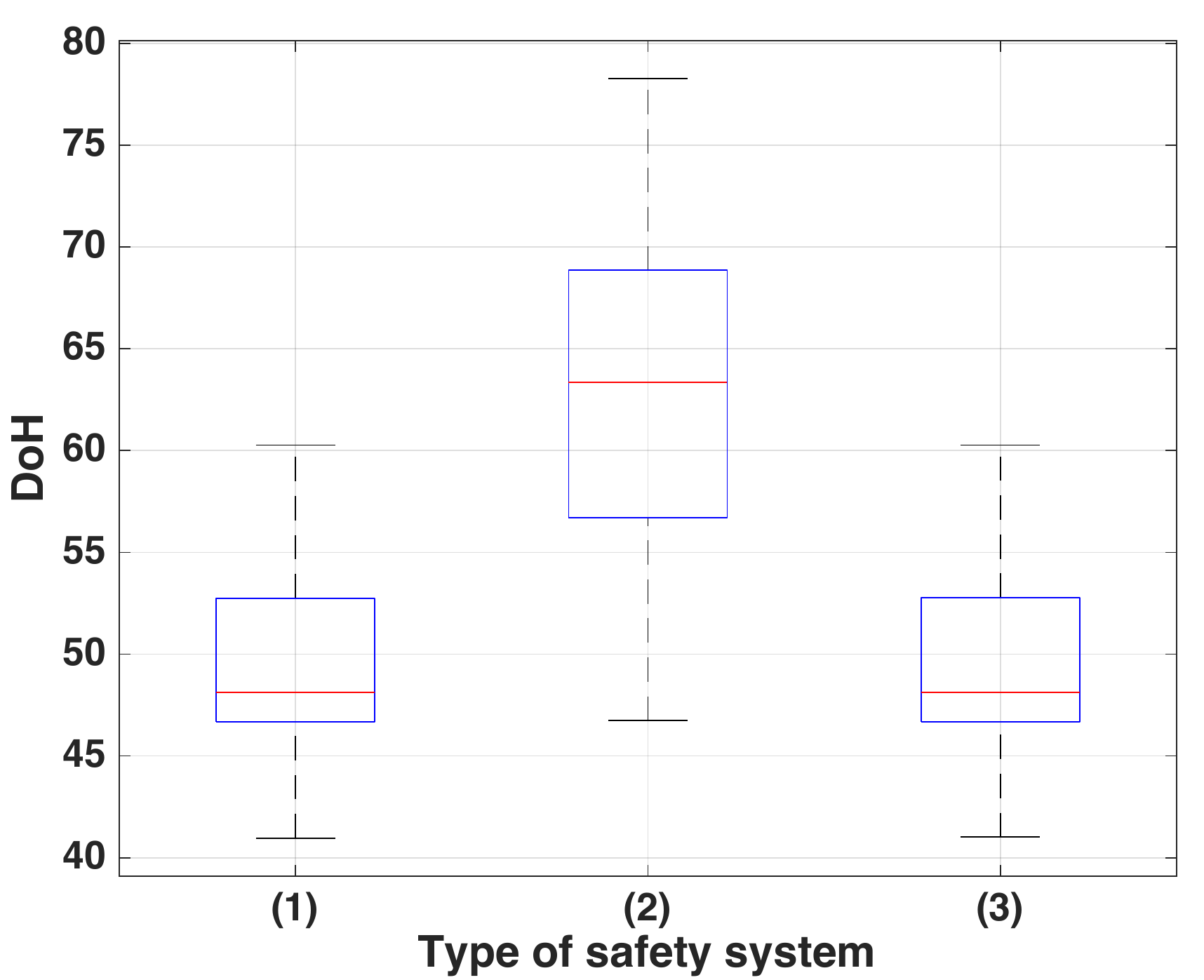}
\caption{Comparing DoH of the patients with different BP safety systems, achieved after 20 min from the start of the closed-loop administration of propofol: (1) the individualized safety system; (2) the model-invariant safety system, (3) the model-invariant safety system including falsification.}
\label{fig:CMPSP-MISP-Fals_20min}
\end{figure}
\end{center}

\vspace{-2cm}
\section{Conclusion}\label{sec:conclusion}
{\color{black}We have proposed a model-invariant safety system for closed-loop anesthesia that guarantees to maintain the patient's blood pressure within a safety bound despite inter-patient variability. 
We have proposed to use model falsification to reduce the conservatism introduced in model-invariant safety-preserving control. 
The effectiveness of blood pressure model falsification is shown using clinical data. The effectiveness of the proposed model-invariant safety-preserving control approach is demonstrated in simulation. Augmenting the model-invariant safety-preserving controller using feedback from noisy blood pressure measurements for model falsification significantly reduces conservatism introduced by the uncertainty. Model-falsification reduces this model uncertainty and subsequently decreases conservatism of the model-invariant safety system.

{\color{black}We have presented this work as a proof-of-concept of a model-invariant safety system for closed-loop anesthesia which includes model-falsification. 
The proposed solution is limited by the assumption that patient models exist in a known, finite model set, and that all outliers are removed from the data prior to falsification.
To generalize the results of this paper to cases in which these assumptions are not met, further research is required.
}
}

\bibliographystyle{unsrt}
\bibliography{myrefs}

\begin{thebibliography}{10}

\bibitem{van2014safety}
K.~van Heusden, N.~West, A.~Umedaly, J.M. Ansermino, R.N. Merchant, and G.A.
  Dumont.
\newblock Safety, constraints and anti-windup in closed-loop anesthesia.
\newblock In {\em IFAC World Congress}, volume~19, pages 6569--6574, 2014.

\bibitem{murphy2011safety}
H.R. Murphy, K.~Kumareswaran, D.~Elleri, J.M. Allen, K.~Caldwell, M.~Biagioni,
  D.~Simmons, D.B. Dunger, M.~Nodale, M.E. Wilinska, and S.A. Amiel.
\newblock Safety and efficacy of 24-h closed-loop insulin delivery in
  well-controlled pregnant women with type 1 diabetes: a randomized crossover
  case series.
\newblock {\em Diabetes care}, 34(12):2527--2529, 2011.

\bibitem{margellos2009decision}
K.~Margellos and J.~Lygeros.
\newblock Air traffic management with target windows: An approach using
  reachability.
\newblock In {\em Proceedings of the 48th IEEE Conference on Decision and
  Control, 2009 held jointly with the 2009 28th Chinese Control Conference.
  CDC/CCC 2009.}, pages 145--150, Dec 2009.

\bibitem{hasenauer2009guaranteed}
J.~Hasenauer, P.~Rumschinski, S.~Waldherr, S.~Borchers, F.~Allg{\"o}wer, and
  R.~Findeisen.
\newblock Guaranteed steady-state bounds for uncertain chemical processes.
\newblock In {\em Proceedings of International Symposium on Advanced Control of
  Chemical Processes, ADCHEM'09, Istanbul, Turkey}, pages 674--679, 2009.

\bibitem{clarke2008birth}
E.M. Clarke.
\newblock The birth of model checking.
\newblock In {\em 25 Years of Model Checking}, pages 1--26. Springer, 2008.

\bibitem{aubin2011viability}
J.P. Aubin, A.M Bayen, and P.~Saint-Pierre.
\newblock {\em Viability theory: new directions}.
\newblock Springer Science \& Business Media, 2011.

\bibitem{lygeros1999controllers}
J.~Lygeros, C.~Tomlin, and S.~Sastry.
\newblock Controllers for reachability specifications for hybrid systems.
\newblock {\em Automatica}, 35(3):349--370, 1999.

\bibitem{kurzhanski1998ellipsoidal}
A.~B. Kurzhanski and P.~Varaiya.
\newblock Ellipsoidal techniques for reachability analysis.
\newblock {\em preprint}, 1998.

\bibitem{gao2007reachability}
Y.~Gao, J.~Lygeros, and M.~Quincampoix.
\newblock On the reachability problem for uncertain hybrid systems.
\newblock {\em IEEE Transactions on Automatic Control}, 52(9):1572--1586, 2007.

\bibitem{kaynama2015scalable}
S.~Kaynama, I.M. Mitchell, M.~Oishi, and G.A. Dumont.
\newblock Scalable safety-preserving robust control synthesis for
  continuous-time linear systems.
\newblock {\em IEEE Transactions on Automatic Control}, 60(11):3065--3070,
  2015.

\bibitem{dumont2009robust}
G.~A. Dumont, A.~Martinez, and J.~M. Ansermino.
\newblock Robust control of depth of anesthesia.
\newblock {\em International Journal of Adaptive Control and Signal
  Processing}, 23(5):435--454, 2009.

\bibitem{khosravi2015constrained}
S.~Khosravi.
\newblock {\em Constrained model predictive control of hypnosis}.
\newblock PhD thesis, University of British Columbia, 2015.

\bibitem{yousefi2016model}
M.~Yousefi, K.~van Heusden, G.A. Dumont, I.M. Mitchell, and J.M. Ansermino.
\newblock Model-invariant safety-preserving control.
\newblock In {\em 2016 American Control Conference (ACC)}, pages 6689--6694.
  IEEE, 2016.

\bibitem{kaynama2012computing}
S.~Kaynama, J.~Maidens, M.~Oishi, I.M. Mitchell, and G.A. Dumont.
\newblock Computing the viability kernel using maximal reachable sets.
\newblock In {\em Proceedings of the 15th ACM international conference on
  Hybrid Systems: Computation and Control}, pages 55--64. ACM, 2012.

\bibitem{maidens2013lagrangian}
J.N. Maidens, S.~Kaynama, I.M. Mitchell, M.K. Oishi, and G.A. Dumont.
\newblock Lagrangian methods for approximating the viability kernel in
  high-dimensional systems.
\newblock {\em Automatica}, 49(7):2017--2029, 2013.

\bibitem{gillulay2011guaranteed}
J.H. Gillula and C.J. Tomlin.
\newblock Guaranteed safe online learning of a bounded system.
\newblock In {\em 2011 IEEE/RSJ International Conference on Intelligent Robots
  and Systems (IROS)}, pages 2979--2984. IEEE, 2011.

\bibitem{gillula2013reducing}
J.H. Gillula and C.J. Tomlin.
\newblock Reducing conservativeness in safety guarantees by learning
  disturbances online: iterated guaranteed safe online learning.
\newblock In {\em Robotics: Science and Systems}, volume~8, page~81, 2013.

\bibitem{PKTKN94}
K.~Poolla, P.~Khargonekar, A.~Tikku, J.~Krause, and K.~Nagpal.
\newblock A time-domain approach to model validation.
\newblock {\em IEEE Transactions on automatic control}, 39(5):951--959, 1994.

\bibitem{yousefi2017modellingl}
M.~Yousefi, K.~van Heusden, J.M. Ansermino, and G.A. Dumont.
\newblock Modelling blood pressure uncertainty for safety verification of
  propofol anesthesia.
\newblock In {\em IEEE International Conference on Systems, Man, and
  Cybernetics (SMC2017)}, pages 1740--1745. IEEE, 2017.

\bibitem{van2018closed}
K.~van Heusden, M.~Yousefi, J.M. Ansermino, and G.A. Dumont.
\newblock Closed-loop miso identification of propofol effect on blood pressure
  and depth of hypnosis.
\newblock {\em IEEE Transactions on Control Systems Technology (Submitted)},
  2018.

\bibitem{bibian2005introduction}
S.~Bibian, C.R. Ries, M.~Huzmezan, and G.A. Dumont.
\newblock Introduction to automated drug delivery in clinical anesthesia.
\newblock {\em European Journal of Control}, 11(6):535--557, 2005.

\bibitem{absalom108overview}
A.~Absalom and M.~Struys.
\newblock {\em An Overview of TCI \& TIVA. 2007, Gent}.
\newblock Academia Press. 1.

\bibitem{schnider1998influence}
T.W. Schnider, C.F. Minto, P.L. Gambus, C.~Andresen, D.B. Goodale, S.L. Shafer,
  and E.J. Youngs.
\newblock The influence of method of administration and covariates on the
  pharmacokinetics of propofol in adult volunteers.
\newblock {\em Anesthesiology}, 88(5):1170--1182, 1998.

\bibitem{eleveld2014general}
D.J. Eleveld, J.H. Proost, L.I. Cort{\'\i}nez, A.R. Absalom, and M.~Struys.
\newblock A general purpose pharmacokinetic model for propofol.
\newblock {\em Anesthesia \& Analgesia}, 118(6):1221--1237, 2014.

\bibitem{kazama1999comparison}
T.~Kazama, K.~Ikeda, K.~Morita, M.~Kikura, M.~Doi, T.~Ikeda, T.~Kurita, and
  Y.~Nakajima.
\newblock Comparison of the effect-site $k_{eo}$'s of propofol for blood
  pressure and {EEG} bispectral index in elderly and younger patients.
\newblock {\em Anesthesiology}, 90(6):1517--1527, 1999.

\bibitem{jeleazcov2015pharmacodynamic}
C.~Jeleazcov, M.~Lavielle, J.~Sch{\"u}ttler, and H.~Ihmsen.
\newblock Pharmacodynamic response modelling of arterial blood pressure in
  adult volunteers during propofol anaesthesia.
\newblock {\em British journal of anaesthesia}, 115(2):213--226, 2015.

\bibitem{gentilini2002}
A.~Gentilini, C.~Schaniel, M.~Morari, C.~Bieniok, R.~Wymann, and T.~Schnider.
\newblock A new paradigm for the closed-loop intraoperative administration of
  analgesics in humans.
\newblock {\em IEEE transactions on biomedical engineering}, 49(4):289--299,
  2002.

\bibitem{peacock1990effect}
J.E. Peacock, R.P. Lewis, C.S. Reilly, and W.S. Nimmo.
\newblock Effect of different rates of infusion of propofol for induction of
  anaesthesia in elderly patients.
\newblock {\em British Journal of Anaesthesia}, 65(3):346--352, 1990.

\bibitem{bibian2006patient}
S.~Bibian, G.A. Dumont, M.~Huzmezan, and C.R. Ries.
\newblock Patient variability and uncertainty quantification in clinical
  anesthesia: Part {I}--{PKPD} modeling and identification.
\newblock In {\em IFAC Symposium on Modelling and Control in Biomedical
  Systems, Reims, France}, 2006.

\bibitem{yousefi2017formally}
M.~Yousefi, K.~van Heusden, I.M. Mitchell, J.M. Ansermino, and G.A. Dumont.
\newblock A formally-verified safety system for closed-loop anesthesia.
\newblock {\em In 20th World IFAC Congress}, 50(1):4424--4429, 2017.

\bibitem{yousefi2018formalized}
M.~Yousefi, K.~van Heusden, I.M. Mitchell, J.M. Ansermino, and G.A. Dumont.
\newblock A formalized safety system for closed-loop anesthesia with
  pharmacokinetic and pharmacodynamic constraints.
\newblock {\em Control Engineering Practice (Submitted)}, 2018.

\bibitem{yousefi2017output}
M.~Yousefi, K.~van Heusden, I.M. Mitchell, and G.A. Dumont.
\newblock Output-feedback safety-preserving control.
\newblock In {\em American Control Conference (ACC), 2017}, pages 2550--2555.
  IEEE, 2017.

\bibitem{schuttler2000population}
J.~Sch{\"u}ttler and H.~Ihmsen.
\newblock Population pharmacokinetics of propofol: a multicenter study.
\newblock {\em Anesthesiology}, 92(3):727--738, 2000.

\bibitem{kuhn2012structure}
T.S. Kuhn.
\newblock {\em The structure of scientific revolutions}.
\newblock University of Chicago press, 2012.

\bibitem{ST97}
M.G. Safonov and T.C. Tsao.
\newblock The unfalsified control concept and learning.
\newblock In {\em Decision and Control, 1994., Proceedings of the 33rd IEEE
  Conference on}, volume~3, pages 2819--2824. IEEE, 1994.

\bibitem{VDS07}
J.~Van~Helvoort, B.~de~Jager, and M.~Steinbuch.
\newblock Direct data-driven recursive controller unfalsification with analytic
  update.
\newblock {\em Automatica}, 43(12):2034--2046, 2007.

\bibitem{Kos01}
R.L. Kosut.
\newblock Uncertainty model unfalsification for robust adaptive control.
\newblock {\em Annual Reviews in Control}, 25:65--76, 2001.

\bibitem{BBMT10}
S.~Baldi, G.~Battistelli, E.~Mosca, and P.~Tesi.
\newblock Multi-model unfalsified adaptive switching supervisory control.
\newblock {\em Automatica}, 46(2):249--259, 2010.

\bibitem{AMM05}
T.~Agnoloni, C.~Manuelli, and E.~Mosca.
\newblock Controller falsification in automatic drug delivery for neuromuscular
  blockade control.
\newblock In {\em 44th IEEE Conference on Decision and Control and European
  Control Conference. CDC-ECC'05}, pages 2990--2995. IEEE, 2005.

\bibitem{margellos2013viable}
K.~Margellos and J.~Lygeros.
\newblock Viable set computation for hybrid systems.
\newblock {\em Nonlinear Analysis: Hybrid Systems}, 10:45--62, 2013.

\bibitem{kurzhanskiui1997ellipsoidal}
A.B. Kurzhanski and I.~V{\'a}lyi.
\newblock {\em Ellipsoidal calculus for estimation and control}.
\newblock Nelson Thornes, 1997.

\bibitem{girard2005reachability}
A.~Girard.
\newblock Reachability of uncertain linear systems using zonotopes.
\newblock In {\em Hybrid Systems: Computation and Control}, pages 291--305.
  Springer, 2005.

\bibitem{abate2008probabilistic}
A.~Abate, M.~Prandini, J.~Lygeros, and S.~Sastry.
\newblock Probabilistic reachability and safety for controlled discrete time
  stochastic hybrid systems.
\newblock {\em Automatica}, 44(11):2724--2734, 2008.

\bibitem{summers2013stochastic}
S.~Summers, M.~Kamgarpour, C.~Tomlin, and J.~Lygeros.
\newblock Stochastic system controller synthesis for reachability
  specifications encoded by random sets.
\newblock {\em Automatica}, 49(9):2906--2910, 2013.

\bibitem{agrawal2010recommended}
G.~Agrawal, S.~Bibian, and T.~Zikov.
\newblock Recommended clinical range for wavcns index during general
  anesthesia.
\newblock In {\em Proceedings of the 2010 Annual Meeting of the American
  Society of Anesthesiologists}, 2010.

\bibitem{zikov2006quantifying}
T.~Zikov, S.~Bibian, G.A. Dumont, M.~Huzmezan, and C.R. Ries.
\newblock Quantifying cortical activity during general anesthesia using wavelet
  analysis.
\newblock {\em IEEE Transactions on Biomedical Engineering}, 53(4):617--632,
  2006.

\bibitem{MPT3}
M.~Herceg, M.~Kvasnica, C.N. Jones, and M.~Morari.
\newblock {Multi-Parametric Toolbox 3.0}.
\newblock In {\em Proceedings of the European Control Conference ({ECC})},
  pages 502--510, Z\"urich, Switzerland, July 17--19 2013.
\newblock {http://control.ee.ethz.ch/~mpt}.

\bibitem{dumont2011closed}
G.~A. Dumont, N.~Liu, C.~Petersen, T.~Chazot, M.~Fischler, and J.~M. Ansermino.
\newblock Closed-loop administration of propofol guided by the neurosense:
  clinical evaluation using robust proportional-integral-derivative design.
\newblock In {\em American Society of Anesthesiologists (ASA) Annual Meeting},
  2011.

\bibitem{lesser2016safety}
K.~Lesser and A.~Abate.
\newblock Safety verification of output feedback controllers for nonlinear
  systems.
\newblock In {\em 2016 European Control Conference (ECC)}. IFAC, 2016.

\end{thebibliography}

\end{document}